\documentclass[letterpaper, DIV=16]{scrartcl}
\usepackage{amssymb, amsthm, mathtools,bbm,bm}
\usepackage[square,sort,comma,numbers]{natbib}
\usepackage{comment}

\usepackage[none]{hyphenat}
\usepackage[english]{babel}
\usepackage{enumerate}
\usepackage{amsmath}
\usepackage{amsthm}
\usepackage{amssymb}
\usepackage{dsfont}%
\usepackage{hyperref}
\usepackage{amssymb}
\usepackage{color}
\usepackage[utf8]{inputenc}
\usepackage{graphicx}
\usepackage{tikz}
\usepackage{tikz-cd}
\usepackage{verbatim} 
	\usepackage{mathtools} 
	\usepackage{mathabx} 
	\usepackage{hyperref}
	\usepackage{comment}
	\usetikzlibrary{decorations.markings}
	\usepackage{bbm} 
	\usepackage{xfrac} 
	\usepackage{tensor} 
	\usepackage{listings}  
	\usepackage{matlab-prettifier} 
	
	\usepackage{pgfplots}              
	\pgfplotsset{compat=newest}
	\usepgfplotslibrary{fillbetween}

	\usepackage[title]{appendix}

\usepackage[shortcuts]{extdash}  
	
	\newtheorem{theorem}{Theorem}[section]  
	\newtheorem{cor}[theorem]{Corollary}
	\newtheorem{prop}[theorem]{Proposition}
	\newtheorem{lemma}[theorem]{Lemma}

	\numberwithin{equation}{section}
	\theoremstyle{definition}
	\newtheorem{defn}{Definition}[section]
	\newtheorem{rem}[theorem]{Remark}
	\newtheorem{assumption}[theorem]{Assumption}

\begin{document}
		\title{Mermin-Wagner theorems for quantum systems\\ with multipole symmetries}
		
		\author{ Timo Feistl$^1$, Severin Schraven$^2$, Simone Warzel$^{2,3,4}$ \\
			\\
            \small $^1$ Department of Mathematics, Friedrich-Alexander-University, Erlangen, Germany \\[-.5ex]
			\small $^2$ Department of Mathematics, TU Munich, Garching, Germany \\[-.5ex]
			\small $^3$ Munich Center for Quantum Science and Technology, Munich, Germany \\[-.5ex]
			\small $^4$ Department of Physics, TU Munich, Garching, Germany}

        \date{\small \today \\[-.5cm]}
		\maketitle

		\begin{abstract}
			We prove Mermin-Wagner-type theorems for quantum lattice systems in the presence of  multipole symmetries. These theorems show that the presence of higher-order symmetries protects against the breaking of lower-order ones. In particular, we prove that the critical dimension in which the charge symmetry can be broken increases if the system admits higher multipole symmetries, e.g. $ d = 4 $ on the regular lattice $ \mathbb{Z}^d $ in the presence of dipole symmetry. \\[.5ex]
		\end{abstract}

		\section{Introduction} \label{sect:Intro}

        The Mermin-Wagner theorem is a milestone in statistical mechanics. It proves that at positive temperature, continuous symmetries cannot be spontaneously broken in spatial dimensions $d\leq 2$. Early manifestations of this principle were established by Mermin and Wagner \cite{MerminWagnerHeisenbergModel} for the rotational symmetry of the isotropic Heisenberg model, and by Hohenberg \cite{hohenberg1967existence} in the context of superfluids. Since then, the theorem has been extended and adapted to many different physical settings and symmetry types~\cite{PfisterPaper, 1981JSP....26..755LLandauPerezGoldstone, KleinMerminWagner, FannesEntropySSB,Simon+1993}. 
        
        In recent years, renewed interest has focused on how additional symmetries modify the dimensional constraints on symmetry breaking. In particular, it has become clear that the coexistence of conventional charge conservation with so-called multipole symmetries can substantially alter the critical dimension at which symmetry breaking becomes possible. Recent work \cite{Stahl2022FieldTheoryMultipole, PhysRevB.106.245125Kapustin} has shown that higher multipole symmetries can protect charge symmetry and thereby raise the dimension required for its spontaneous breaking. The purpose of this work is to prove that this mechanism is not specific to charge conservation: rather, higher-order multipole symmetries generically protect lower-order ones. Moreover, we also eliminate  some additional assumptions required in the proof strategy of the Mermin-Wagner results in~\cite{PhysRevB.106.245125Kapustin}. 

        Multipole symmetries are not merely of abstract interest but arise naturally in several physically relevant systems. 
        One class of examples are fractional quantum Hall models, such as those based on Haldane pseudopotentials \cite{PhysRevLett.51.605HaldaneFQHE, PhysRevB.88.165303SeidlFQHE, PhysRevX.5.041003LeeGeometricConstruction} or their truncation \cite{nachtergaele2021spectral,warzel2023bulk}. These exhibit both charge and dipole conservation, which, in this case, stems from the conservation of an appropriate angular momentum. Lattice models with dipole symmetries have also started to become directly experimentally accessible via cold atoms in tilted optical lattices \cite{PhysRevX.10.011042SubdiffusionTiltedFermiHubbard, Scherg:2020mcp}. 
        Aside from a change in the critical dimension for symmetry breaking, on which we focus in this paper, 
        the presence of such a dipole symmetry causes other interesting phenomena, such as the appearance of immobile, low-energy excitations known as fractons, of many-body scars high up in the energy spectrum, as well as the change of hydrodynamic laws, see e.g.~\cite{Pollmann,Pretko:2020aa,Morningstar:2020aa,Bernevig10.19,Gorantla:2022aa,Feldmeier:2020aa,PhysRevLett.132.137102,PhysRevB.107.195131,HHY34,ZSPK24}.

		\subsection{Setting} \label{sect:setting}
		
		We formulate our results in the standard setting for quantum 'lattice' systems. Instead of a regular lattice, we work on a countable metric space $(\mathcal{L}, D)$, which is embedded in $ \mathbb{R}^d $ and whose volume growth defines an effective dimension. 
        \begin{assumption} (Base space)
			\label{as: lattice}
			Let $\mathcal{L}$ be a countable subset of $\mathbb{R}^d$ for some spatial dimension $d\in\mathbb{N}$ and  $D$ be the metric on $\mathcal{L}$ induced by the standard Euclidean metric $\vert \cdot \vert$ on $\mathbb{R}^d$. We assume that there exist constants $\gamma>0$, $C\geq 0$ and $r_0\geq 0$ such that for balls $ B_r(z) \coloneqq \{ x \in \mathbb{R}^d \ \vert \ |x-z| < r \} $
			\begin{equation}
				\label{eq: geometric dimension condition}
				\vert B_r(z)\cap\mathcal{L}\vert\leq C  r^\gamma \qquad \text{for all } r\geq r_0 \text{ and all } z\in\mathcal{L}.
			\end{equation}
            We refer to $ \gamma $ as the effective dimension of $\mathcal{L} $. 
		\end{assumption}
        The most important example will be $\mathcal{L} =\mathbb{Z}^d $ for which $ \gamma = d $. \\
        The set of finite, nonempty subsets of $\mathcal{L}$ will be denoted by $P_0(\mathcal{L})$. To each site $x\in \mathcal{L}$, we associate a finite-dimensional complex Hilbert space $\mathcal{H}_x$. The quantum mechanics of the degrees of freedom supported on $\Lambda\in P_0(\mathcal{L})$ is then described by the $C^*$-algebra $\mathcal{U}_\Lambda$ of bounded linear operators $ L(\mathcal{H}_\Lambda) $ on the Hilbert-space tensor product $\mathcal{H}_\Lambda=\bigotimes_{x\in \Lambda}\mathcal{H}_x$. \\
		  For any nested pair $\Lambda_1,\Lambda_2\in P_0(\mathcal{L})$, $\Lambda_1\subseteq\Lambda_2$, the map 
		\begin{equation}
			\label{eq: natural inclusion}
			\mathcal{U}_{\Lambda_1}\ni A\mapsto A\otimes\bigl(\bigotimes_{x\in\Lambda_2\setminus\Lambda_1}\mathbb{I}_{\mathcal{H}_x}\bigr)\in \mathcal{U}_{\Lambda_2}
		\end{equation}
		provides a natural identification under which $\mathcal{U}_{\Lambda_1}\subseteq\mathcal{U}_{\Lambda_2}$. 
		The algebra of local operators is the following union respecting inclusions of the form \eqref{eq: natural inclusion} 
		\begin{equation}
			\mathcal{U}_{\mathrm{loc}}\coloneqq\bigcup_{\Lambda\in P_0(\mathcal{L})}\mathcal{U}_{\Lambda}.
		\end{equation}
		When equipping $\mathcal{U}_{\mathrm{loc}}$ with the norm $ \| \cdot \| $ inherited from the operator norm on $L(\mathcal{H}_\Lambda)$ for $\Lambda\in P_0(\mathcal{L})$, 
		the norm completion of $\mathcal{U}_{\mathrm{loc}}$ is the $C^*$-algebra $\mathcal{U}$ of quasi-local operators on $(\mathcal{L},D)$. 
		The unit of $\mathcal{U}$ is denoted by $\mathbb{I}$. For not necessarily finite subsets $X\subseteq \mathcal{L}$, we also define $\mathcal{U}_X$ to be the norm closure of $\bigcup_{\substack{\Lambda\in P_0(\mathcal{L})},\Lambda\subseteq X}\mathcal{U}_\Lambda$.
		For a more detailed exposition of algebras of quasi-local operators as well as proofs of the relevant properties, see \cite[Chapter 6.2]{BratelliRobinson2} or \cite[Chapter II.1]{Simon+1993}. 
		
		The possible physical states of the system are modeled by states of the $C^*$-algebra $\mathcal{U}$, i.e., by positive linear functionals $\omega:\mathcal{U}\rightarrow\mathbb{C}$ with $\omega(\mathbb{I})=1$. Such functionals are automatically continuous.
	
		A time-evolution on $\mathcal{U}$ is a strongly continuous 1-parameter group of $*$-automorphisms $\alpha=(\alpha_t)_{t\in\mathbb{R}}$. Given a countable metric space $(\mathcal{L},D)$ with its associated $C^*$-algebra of quasi-local operators $\mathcal{U}$, we will call a family $\Phi=\bigl(\phi(\Lambda)\bigr)_{\Lambda\in P_0(\mathcal{L})}$ an interaction if all $\phi(\Lambda)$ are self-adjoint elements of $\mathcal{U}_\Lambda$. We consider time-evolutions $\alpha$ which are  induced by an interaction $\Phi$ in such a way that $\phi(\Lambda)$ can be interpreted as a local Hamiltonian. Informally, this is the Heisenberg evolution 
        \begin{equation} \label{alpha}
			\alpha^\Phi_t(A) = \exp\left(it \sum_{\Lambda \in P_0(\mathcal{L})} \phi(\Lambda)\right) A\exp\left(-it \sum_{\Lambda \in P_0(\mathcal{L})} \phi(\Lambda)\right) . 
		\end{equation}
        Under suitable decay conditions on the interactions, which we spell below, one obtains such a time-evolution for all $A \in \mathcal{U}$, see Appendix \ref{sect:Dynamics}.

		Our Mermin-Wagner theorem addresses thermal equilibrium states at non-zero temperature. These are modeled by the so-called Kubo-Martin-Schwinger (KMS) condition. To define such KMS-states, one extends the dynamics given by the group of automorphisms $\alpha=(\alpha_t)_{t\in \mathbb{R}}$ to imaginary times. In general, this is not possible, however, if $A\in\mathcal{U}$ is such that the map $\mathbb{R}\ni t\mapsto \alpha_t(A)$ has an extension to an analytic function $\mathbb{C}\ni z\mapsto \alpha_z(A)$, then we call $A$ an $\alpha$-analytic element. We will denote by $\mathcal{U}_\alpha\subseteq\mathcal{U}$ the $*$-subalgebra of $\alpha$-analytic elements.


		\begin{defn} (KMS states) 
			\label{def: KMS states}
			Let $\alpha=(\alpha_t)_{t\in\mathbb{R}}$ be a strongly continuous 1-parameter group of automorphisms of $\mathcal{U}$. A state $\omega$ on $\mathcal{U}$ is called a KMS state of $(\mathcal{U},\alpha)$ at inverse temperature $\beta$ if 
			\begin{equation}
				\omega\bigl(A\alpha_{i\beta}(B)\bigr)=\omega(BA)
			\end{equation}
			for all $A,B$ in $\mathcal{U}_\alpha$. We denote the set of $(\mathcal{U},\alpha)$-KMS states at $\beta\in\mathbb{R}$ by $K_\beta(\alpha)$.
		\end{defn}
		One can show that this KMS condition captures the defining properties of thermal equilibrium (see \cite[ Section 5.3.1]{BratelliRobinson2}).  For $C^*$-algebras of quasi-local operators and time-evolutions induced by interactions, the KMS condition is closely related to a maximum entropy principle (see \cite[Section 6.2.3]{BratelliRobinson2}). The relative entropy will play a crucial role in our proof, see Appendix~\ref{sect:PfisterFröhlich}. 
		
		
		The last ingredient in the formulation of the Mermin-Wagner theorem is that of a symmetry of the time-evolution and that of symmetry-preserved or broken states. 
		\begin{defn}  (Symmetries and their breaking)
			\label{def: symmetries}
			Let $\alpha=(\alpha_t)_{t\in\mathbb{R}}$ be a strongly continuous 1-parameter group of automorphisms of $\mathcal{U}$.
			An automorphism $\tau\in Aut(\mathcal{U})$ is called a symmetry of $(\mathcal{U},\alpha)$ if 
			\begin{equation*}
				\alpha_t\circ\tau=\tau\circ\alpha_t \;\; \mathrm{for}\;\mathrm{all}\;t\in\mathbb{R}.
			\end{equation*}
			If $\omega$ is a state on $\mathcal{U}$, then a symmetry $\tau$ of $(\mathcal{U}, \alpha)$ is said to be preserved by $\omega$ if 
			\begin{equation*}
				\omega\circ\tau=\omega.
			\end{equation*}
			Otherwise, $\tau$ is said to be broken on $\omega$. 
		\end{defn}
		One can show 
		that if $\omega\in K_\beta(\alpha)$ and $\tau$ is a symmetry of $(\mathcal{U},\alpha)$, then $\omega \circ\tau\in K_\beta(\alpha)$. Thus, this definition captures the intuitive picture that pre-composing with a symmetry can at most swap between thermal equilibrium states at the same $\beta$. In light of this, the question of whether a given symmetry is spontaneously broken on a given KMS state is natural. 
		
		\subsection{Multipole symmetries and main result}

		Our main theorem will address symmetries generated by multipole moments for a family of local charge operators. 
		
		\begin{assumption} (Family of local charge operators) 
			\label{as: family of charge operators}
            We call a family  $\{n_x \ \vert \ x\in\mathcal{L}\}\subseteq\mathcal{U}_{\mathrm{loc}}$ of self-adjoint operators $n_z$ a family of local charge operators if 
			\begin{enumerate}[(i)]
				\item $[n_x,n_y]=0$ for all $x,y\in\mathcal{L}$.
				\item There exists $N_0>0$ such that $\Vert n_x \Vert\leq N_0$ for all $x\in \mathcal{L}$.
				\item There exists  $R_0>0$ such that the support satisfies $\mathrm{supp}(n_x)\subseteq B_{R_0}(x) \cap \mathcal{L} $ for all $x\in \mathcal{L}$. 
			\end{enumerate}
		\end{assumption}

		The operator $n_x$ models a local charge operator near site $x\in\mathcal{L}$. Conditions $(i)$ and $(iii)$ are automatically satisfied if $n_x\in \mathcal{U}_{\{x\}}$ for all $x\in \mathcal{L}$, which is the case for many typical models of quantum statistical mechanics, cf.~\cite{Simon+1993,naaijkens2017quantum,LectureNotes}.
		
		The multipole symmetry associated with some multi-index $a\in \mathbb{N}_0^d$ is defined via the following 1-parameter group of automorphisms
		\begin{equation}
			\label{eq:symmetries}
			\tau^{(a)}_s(A)=\exp{\Bigl(-is\sum_{\substack{x\in\mathcal{L}:\\\mathrm{supp}(n_x)\cap\Lambda\ne\emptyset}}x^an_x\Bigr)}A\exp{\Bigl(is\sum_{\substack{x\in\mathcal{L}:\\\mathrm{supp}(n_x)\cap\Lambda\ne\emptyset}}x^an_x\Bigr)} , \qquad A \in \mathcal{U}_\Lambda , 
		\end{equation}
		with $x^a = \prod_{j=1}^d x_j^{a_j}$ and $ s \in \mathbb{R} $. Equation \eqref{eq:symmetries} defines $ \tau^{(a)}_s $ on all local observables $A\in \mathcal{U}_\mathrm{loc}$. We show in Proposition \ref{th: definition of the multipole symmetries} that one can extend it to a strongly continuous one-parameter group of automorphisms on the entire quasi-local algebra.
        For example, in case $ a = 0  $ the generator of the symmetry is the total charge. In case $a =  (1, 0 , \dots , 0) = e_1 $, the dipole moment along the first coordinate axis generates the symmetry.  
		
		As a last assumption, we require that interactions decay sufficiently fast and are invariant under the multipole symmetries up to some order. 
		\begin{assumption} (Symmetric interactions) \label{as: interaction}
			We assume that our interaction $\Phi=(\phi(\Lambda))_{\Lambda\in P_0(\mathcal{L})}$ is $ k $-symmetric with $k\in \mathbb{N}_0$ in the  sense that:
            \begin{enumerate}
                \item for every multi-index $a\in \mathbb{N}_0^d$ with $\vert a\vert \coloneqq \sum_{j=1}^d a_j \leq k$, we have 
			\begin{equation} \label{as:invariance}
				\tau_s^{(a)}\bigl(\phi(\Lambda)\bigr)=\phi(\Lambda)
			\end{equation}
			for all $\Lambda\in P_0(\mathcal{L})$ and all $s\in\mathbb{R}$. 
                \item we have 
			\begin{equation} \label{decay condition}
				\sup_{x\in\mathcal{L}}\Biggl(\sum_{\substack{\Lambda\in P_0(\mathcal{L}):\\ x\in\Lambda}}\Vert\phi(\Lambda)\Vert\Bigl(\sum_{\substack{z\in\mathcal{L}: \\ \mathrm{supp}(n_z)\cap\Lambda\ne\emptyset}}\vert z-x\vert^{k+1}\Bigr)^2\Biggr)<\infty.
			\end{equation}
            \item  there exists $\varepsilon>0$ such that 
			\begin{equation} \label{decay condition 2}
				\sup_{x,y\in\mathcal{L}}(1+\vert x -y\vert)^{\gamma+1+\varepsilon}\sum_{\Lambda\in P_0(\mathcal{L}): \;x,y\in \Lambda}\Vert\phi(\Lambda)\Vert <\infty,
			\end{equation}
            with $ \gamma $ the effective dimension of $ \mathcal{L} $. 
            \end{enumerate}
		\end{assumption}
        For each $ a \in \mathbb{N}_0^d$, the symmetry $(\tau_s^{(a)})_{s\in\mathbb{R}}$ is continuous in the usual jargon of statistical mechanics. 
        The symmetry condition \eqref{as:invariance} is in general strictly stronger than Definition \ref{def: symmetries}. Generically, many different interactions induce the same time-evolution $\alpha$, and whether this stronger symmetry condition holds or not can depend on the choice of the interaction.\\
		The decay conditions serve two different purposes: Condition \eqref{decay condition 2} ensures that the interaction  $\Phi$ induces a time-evolution $ \alpha $ on $\mathcal{U}$ in the sense of \eqref{alpha}. This is proved in Appendix \ref{sect:Dynamics} using standard tools. Condition \eqref{decay condition} is needed to control the second-order error in the Taylor expansion needed in the proof of the Mermin-Wagner result. The rate of decay \eqref{decay condition}  hence depends on the order of the multipole symmetry $k\in \mathbb{N}_0$ under consideration. We will show in Corollary~\ref{cor: suff. cond for decay cond.} in the appendix that
		\begin{equation} \label{simpler decay}
			\sup_{x,y\in \mathcal{L}} (1+\vert x -y\vert)^{4+2k+2\gamma} \sum_{\Lambda\in P_0(\mathcal{L}) \ : \ x,y\in \Lambda} \Vert \phi(\Lambda) \Vert <\infty,
		\end{equation} 
		is sufficient for condition \eqref{decay condition}. Hence  sufficiently fast algebraic decay of the interaction ensures the validity of both  \eqref{decay condition} and \eqref{decay condition 2}. \\

        Our main result is the following collection of Mermin-Wagner theorems on the absence of breaking of the multipole symmetries in sufficiently low effective dimension. 
		
		\begin{theorem} 
			\label{Thm:Main}
		In the setting of Assumptions \ref{as: lattice} and \ref{as: family of charge operators}, suppose that Assumption \ref{as: interaction} holds for $ k \in \mathbb{N}_0 $. If $a\in \mathbb{N}_0^d$ is a fixed multi-index with $\vert a \vert \leq k$ and the effective dimension satisfies $\gamma \leq 2(k-\vert a \vert +1)$, then the multipole symmetry $(\tau_s^{(a)})_{s\in\mathbb{R}}$ is preserved on every $\beta$-KMS state $\omega$ for all $\beta\in\mathbb{R}$; i.e., $\omega\circ\tau_s^{(a)}=\omega$ for all $s\in\mathbb{R}$ and $\omega\in K_\beta(\alpha)$. 
		\end{theorem}

		Let us demonstrate the applicability of the result for $\mathcal{L}=\mathbb{Z}^d$, i.e., when $d=\gamma$. The case $ k= 0 $ covers the absence of the breaking of the charge symmetry ($a=0$) if $d=\gamma\leq 2 $. We hence recover the standard Mermin-Wagner theorem. In case $ k = 1 $, both charge and the dipole symmetries in all directions are present. In this case, Theorem~\ref{Thm:Main} rules out the breaking of the charge symmetry ($ a= 0 $) in case $d=\gamma\leq 4 = 2 (1-0+1)$. In contrast, the dipole symmetries associated with the $ d $ directions ($a = e_j$ with $ j \in \{ 1, \dots . d\} $) remain unbroken only for $d= \gamma \leq 2 $. The case $ k = 1 $ thus constitutes a minor generalization of the Mermin-Wagner theorem in \cite{PhysRevB.106.245125Kapustin}, where Theorem~\ref{Thm:Main} shows that one may drop the additional assumption on clustering of the state.

		\subsection{Comments and outline of the paper}
        Theorem~\ref{Thm:Main} is proven in Section \ref{Sect:Proof}. As a preliminary step, we show that the multipole symmetry exists in the infinite-volume limit, ensuring that the problem is well-posed. Using standard methods, we then reduce the Mermin–Wagner theorem to a bound on the relative entropy of the KMS state $ \omega $ and its symmetry-transformed KMS state $ \omega \circ \tau_s^{(a)}$. As usual, this bound is proven to be finite via an appropriate Taylor expansion of a smoothly truncated unitary $U_m^{(a)}(s) $ representing the multipole symmetry.

Appendix \ref{sect:Dynamics} reviews standard arguments involving general $ F$-functions (Definition~\ref{def: F functions}), which establish the existence of the infinite-volume time-evolution $ \alpha $ generated by our interactions. Appendix \ref{sect:PfisterFröhlich} reviews the proof of the entropy-based criterion from \cite{PfisterPaper}. We also discuss a variant of the main theorem for slab geometries of the form $\mathcal{L}\subset \mathbb{R}^{d-\ell}\times[-M,M]^\ell$ in Appendix~\ref{sect:Slab}.

      Our aim was not to present the most general form of the Mermin-Wagner theorem for multipole symmetries, but rather a version that applies to physically relevant systems and relies on assumptions that are easy to verify. For instance, Assumption~\ref{as: lattice} could be weakened by only requiring  that $\lvert \mathcal{L}\cap B_r(x)\rvert$ is finite for all $r>0$ and $x\in\mathcal{L}$, provided that \eqref{decay condition 2} is replaced by the existence of an $F$-function for $\mathcal{L}$. We impose Assumption~\ref{as: lattice} mainly to ensure that $F(r)=(1+r)^{1+\gamma+\varepsilon}$ is an $F$-function for $\mathcal{L}$ and to guarantee the regularity of the truncated symmetry $U_m^{(a)}$; precise definitions can be found in Appendix \ref{sect:Dynamics}. We emphasize that we do not assume any form of clustering, which is the main advantage compared with the strategy of \cite{PhysRevB.106.245125Kapustin}.
		
		\section{Proof of Theorem \ref{Thm:Main}} \label{Sect:Proof}
		The proof proceeds in three steps. We first establish the existence of the multipole symmetry $ \tau_s^{(a)} $ as a continuous symmetry on $ \mathcal{U} $. We then employ the entropy methods from \cite{PfisterPaper} and reduce the proof of Theorem~\ref{Thm:Main} to the following finiteness criterion. Its formulation requires the notion of exhausting sequences, which we will use throughout the proof. 
        \begin{defn} (Exhausting sequences of finite subsets) We write $ \Lambda_n \uparrow \Lambda $ for any sequence of finite sets $ \Lambda_n\subseteq\mathcal{L}$, which is increasing $ \Lambda_n \subset \Lambda_{n+1} $ and which exhausts the entire space, $ \bigcup_{n\in \mathbb{N}} \Lambda_n = \mathcal{L} $.
        \end{defn}
        The following is an adaptation of~\cite{PfisterPaper} (see also~\cite{NachtergaeleMerminWagner}) to the present set-up. 
				\begin{prop} \label{th: Proposition 1}
			In the setting of Theorem~\ref{Thm:Main} suppose that for some $\Lambda_m\uparrow\mathcal{L}$ with an associated sequence of unitaries $(U_m)_{m\in\mathbb{N}}\subseteq\mathcal{U}_{\mathrm{loc}}$ such that $ \tau^{(a)}_s(A)=U_m^*AU_m $ for all $m\in\mathbb{N}$ and $ A\in\mathcal{U}_{\Lambda_m} $ 
            we have
            \begin{equation}
					\label{eq: bound expectation value of U_m delta U_m^* not appendix}
                              \sup_{m\in\mathbb{N}} \ \bigl\|U_m\delta(U_m^*)+U_m^*\delta(U_m\bigr)\bigr\|<\infty,
				\end{equation}
                with $ \delta $ the infinitesimal generator of the time-evolution $ \alpha $. 
				Then $\omega\circ\tau_s^{(a)} = \omega$ for all $\omega\in K_\beta(\alpha)$; i.e., $\tau_s^{(a)}$ cannot be spontaneously broken on any thermal equilibrium states at inverse temperature $\beta$. 
            \end{prop}
		As we recall in Appendix \ref{sect:PfisterFröhlich}, where the proof of this criterion can be found, the expression in~\eqref{eq: bound expectation value of U_m delta U_m^* not appendix} serves as an upper bound on the relative entropy of the state $ \omega $ and $ \omega \circ \tau_s^{(a)}$.  As standard in Mermin-Wagner arguments, the  finite-volume approximations $ U_m $ of the symmetry $ \tau_s^{(a)} $ are constructed in terms of smooth truncations, see~\eqref{eq: definition of U_m} below. 
		We will then show that \eqref{eq: bound expectation value of U_m delta U_m^* not appendix} is indeed satisfied for those truncations.


		\subsection{Smooth truncations and existence of the multipole symmetry} \label{sect:symmetry}
		
		In this subsection, we will prove that $\tau^{(a)}_s$ defined in \eqref{eq:symmetries} can be extended to a symmetry. Moreover, these symmetries can be approximated by smooth locally generated unitaries. 
		To implement the approximations, we use a smooth characteristic function of the box $[-m,m]^d\subseteq \mathbb{R}^d$ for $m\in\mathbb{N}$. More precisely,  
		we pick a cut-off function $\tilde{\chi}\in C_c^\infty (\mathbb{R})$ such that $0\leq \tilde{\chi}\leq 1$, $\tilde{\chi}(x)=1$ if $| x |\leq 1$, and $\tilde{\chi}(x)=0$ if $| x | \geq 2$. We then set
		$$\chi:\mathbb{R}^d \rightarrow \mathbb{R}, x\mapsto \prod_{j=1}^d \tilde{\chi}(x_j).$$
		This yields $\chi\in C_c^\infty(\mathbb{R}^d)$ with $0\leq \chi\leq 1$, $\chi(x)=1$ if $ | x |_\infty \coloneqq \max_{j\in \{1,\dots, d\}} \vert x_j\vert \leq 1$, and $\chi(x)=0$ if $ | x |_\infty \geq 2$.
		Lastly, for $m\in\mathbb{N}$ we define the rescaled version $$\chi_m:\mathbb{R}^d \rightarrow \mathbb{R}, \chi_m(x) = \chi(x/m),$$ 
		which defines $\chi_m\in C_c^\infty(\mathbb{R}^d)$ with $\chi_m(x)=0$ if $ | x |_\infty \geq 2m$.
		For $m\in\mathbb{N}$, each multi-index $a=(a_1, ..., a_d)\in\mathbb{N}_0^d$, and $s\in\mathbb{R}$ we then set 
		\begin{equation}
			\label{eq: definition of U_m}
			U_m^{(a)}(s):=\exp{\Bigl(is\sum_{x\in\mathcal{L}}\chi_m(x)x^an_x\Bigr)}.
		\end{equation}
		By Assumptions~\ref{as: lattice} and \ref{as: family of charge operators} the series in the exponential
		is a finite sum. As the charge operators are  self-adjoint and $x^a = \prod_{j=1}^d x_j^{a_j}\in \mathbb{R}$, this implies that $U_m^{(a)}(s)$ is a unitary element of $\mathcal{U}_{\mathrm{loc}}$. 
        The following proposition shows that these unitaries indeed approximate the multipole symmetries, which act as~\eqref{eq:symmetries} on local observables. 
		\begin{prop}
						\label{th: definition of the multipole symmetries}
			In the setting of Theorem~\ref{Thm:Main} let $a\in\mathbb{N}_0^d$ be an arbitrary multi-index.
			\begin{itemize}
				\item[(i)]  For every $A\in\mathcal{U}$ and $s\in\mathbb{R}$, the limit $\tau_s^{(a)}(A)=\lim_{m\rightarrow\infty}U_m^{(a)}(s)^*AU_m^{(a)}(s)$ exists. It defines a strongly continuous, norm-preserving 1-parameter group $\tau^{(a)}= (\tau_s^{(a)})_{s\in\mathbb{R}}$ of $*$-automorphisms of $\mathcal{U}$.
				\item[(ii)] For $\Lambda\in P_0(\mathcal{L})$ and $A\in\mathcal{U}_{\Lambda}$, Eq.~\eqref{eq:symmetries} holds.
			\end{itemize}
		\end{prop}

		\begin{proof}
			For $\Lambda\in P_0(\mathcal{L})$ there exists $m_0\in\mathbb{N}$ large enough such that $\vert z\vert\leq m_0-R_0$ for all $z\in\Lambda$, with $R_0$ given in Assumption \eqref{as: family of charge operators}. By definition, $\chi_m(z)=1$ for all $z\in\Lambda$ and $m\geq m_0$. Since the $n_z$ pairwise commute by assumption, this implies for all $m\geq m_0$ and $A\in\mathcal{U}_\Lambda$:
			\begin{equation*}
				U_m^{(a)}(s)^*AU_m^{(a)}(s)=\exp{\Bigl(-is\sum_{\substack{z\in\mathcal{L}:\\\mathrm{supp}(n_z)\cap\Lambda\ne\emptyset}}z^an_z\Bigr)}A\exp{\Bigl(is\sum_{\substack{z\in\mathcal{L}:\\\mathrm{supp}(n_z)\cap\Lambda\ne\emptyset}}z^an_z\Bigr)}.
			\end{equation*} 
			Here we have used that $A\in \mathcal{U}_\Lambda$ commutes with $n_z$ if $\mathrm{supp}(n_z)$ is disjoint from $\mathrm{supp}(A)=\Lambda$. Furthermore, we have used that $\chi_m(z)=1$ for all $z\in \mathcal{L}$ such that $\mathrm{supp}(n_z)\cap \Lambda \neq \emptyset$ as $\mathrm{supp}(n_z)\subseteq B_{R_0}(z)$.
			The above computation shows that the sequence $\bigl(U_m^{(a)}(s)^*AU_m^{(a)}(s)\bigr)_m$ is constant for $m\geq m_0$ and hence in particular $\lim_{m\rightarrow\infty}U_m^{(a)}(s)^*AU_m^{(a)}(s)$ exists and is of the form claimed in \eqref{eq:symmetries}. 

            Using the fact that $U_m^{(a)}(s)$ is unitary and that the local algebra is dense in the quasi-local algebra, one obtains that $\bigl(U_m^{(a)}(s)^*AU_m^{(a)}(s)\bigr)_m$ is a Cauchy sequence for any $ A \in \mathcal{U}$ and thus convergent in $\mathcal{U}$. 
			Moreover, $\tau_s^{(a)} $ when restricted to $ \mathcal{U}_{\mathrm{loc}}$ is clearly linear, multiplicative, commutes with taking the adjoint, is a strongly continuous 1-parameter group, and norm preserving, $  \| \tau_s^{(a)}(A) \| = \| A \| $. These properties extend to $\mathcal{U}$ by continuity, which is uniform in $ s $ on compact subsets. Hence, the $\tau_s^{(a)}$ form a strongly continuous 1-parameter group of $*$-automorphisms.  
		\end{proof}
		The $1$-parameter group of $*$-automorphisms $\bigl(\tau_s^{(a)}\bigr)_{s\in\mathbb{R}}$ describes the multipole-symmetry associated with the multi-index $a$; in particular the charge symmetry corresponds to the case $a=0$. 
		We will use the following sufficient condition for when these multipole automorphism groups $\tau^{(a)}$ are symmetries of the $C^*$-dynamical system $(\mathcal{U}, \alpha)$ in the sense of Definition \ref{def: symmetries}. 
		\begin{prop}
						\label{th: stronger symmetry assumption}
			In the setting of Theorem~\ref{Thm:Main} with $k\in \mathbb{N}_0$:
            \begin{enumerate}
                \item The interaction $\Phi$ induces a time-evolution $\alpha$.

                \item For any multi-index $a\in\mathbb{N}_0^d$ with $\vert a \vert \leq k$, 
			$\tau^{(a)}$ is a symmetry of $(\mathcal{U},\alpha)$, i.e., $\alpha_t\circ\tau_s^{(a)}=\tau_s^{(a)}\circ\alpha_t$ for all $t,s\in\mathbb{R}$. 
            \end{enumerate}
		\end{prop}
		\begin{proof}
        Assumptions \eqref{as: lattice} and \eqref{as: family of charge operators} imply that $F(r)=(1+r)^{1+\gamma+\varepsilon}$ is an $F$-function for $\Phi$ and therefore, the corresponding time-evolution exists thanks to standard results summarized in Proposition~\ref{th: time evolution}.

			Using the symmetry-preservation~\eqref{as:invariance}, Proposition~\ref{th: definition of the multipole symmetries}, and the assumption that the $n_x$ pairwise commute,  for $\Lambda_m \uparrow \mathcal{L}$ and all $\Lambda\subseteq\Lambda_m$ and $A\in\mathcal{U}_\Lambda$ we have
			\begin{align}
				\label{eq: stronger symmetry assumption}
					\notag \phi(\Lambda)&=\tau_s^{(a)}(\phi(\Lambda))=\exp{\Bigl(-is\sum_{\substack{x\in\mathcal{L}:\\\mathrm{supp}(n_x)\cap \Lambda\ne\emptyset}}z^an_z\Bigr)}\phi(\Lambda)\exp{\Bigl(is\sum_{\substack{x\in\mathcal{L}:\\\mathrm{supp}(n_x)\cap \Lambda\ne\emptyset}}x^an_x\Bigr)}\\
					&=\exp{\Bigl(-is\sum_{\substack{x\in\mathcal{L}:\\\mathrm{supp}(n_x)\cap \Lambda_m\ne\emptyset}}x^an_x\Bigr)}\phi(\Lambda)\exp{\Bigl(is\sum_{\substack{x\in\mathcal{L}:\\\mathrm{supp}(n_x)\cap \Lambda_m\ne\emptyset}}x^an_x\Bigr)}.
			\end{align}
			For $m$ large enough, we use \eqref{eq:symmetries} and \eqref{eq: stronger symmetry assumption} to rewrite
            \begin{align}
    \notag &\exp{\Bigl(it\sum_{\Lambda\subseteq\Lambda_m}\phi(\Lambda)\Bigr)}\tau_s^{(a)}(A)\exp{\Bigl(-it\sum_{\Lambda\subseteq\Lambda_m}\phi(\Lambda)\Bigr)} \\ 
    \notag 
&=\exp{\Bigl(it\sum_{\Lambda\subseteq\Lambda_m}\phi(\Lambda)\Bigr)}\exp{\Bigl(-is\smashoperator[l]{\sum_{\substack{x\in\mathcal{L}:\\\mathrm{supp}(n_x)\cap\Lambda_m\ne\emptyset}}}x^an_x\Bigr)}A\exp{\Bigl(is\smashoperator[l]{\sum_{\substack{x\in\mathcal{L}:\\\mathrm{supp}(n_x)\cap\Lambda_m\ne\emptyset}}}x^an_x\Bigr)}\exp{\Bigl(-it\sum_{\Lambda\subseteq\Lambda_m}\phi(\Lambda)\Bigr)} \\
                \notag 
                &=\exp{\Bigl(-is\smashoperator[l]{\sum_{\substack{x\in\mathcal{L}:\\\mathrm{supp}(n_x)\cap\Lambda_m\ne\emptyset}}}x^an_x\Bigr)} \exp{\Bigl(it\sum_{\Lambda\subseteq\Lambda_m}\phi(\Lambda)\Bigr)}A\exp{\Bigl(-it\sum_{\Lambda\subseteq\Lambda_m}\phi(\Lambda)\Bigr)} \exp{\Bigl(is\smashoperator[l]{\sum_{\substack{x\in\mathcal{L}:\\\mathrm{supp}(n_x)\cap\Lambda_m\ne\emptyset}}}x^an_x\Bigr)}\\
                &=\tau_s^{(a)}\left( \exp{\Bigl(it\sum_{\Lambda\subseteq\Lambda_m}\phi(\Lambda)\Bigr)}A\exp{\Bigl(-it\sum_{\Lambda\subseteq\Lambda_m}\phi(\Lambda)\Bigr)}  \right).
            \end{align}
           By Proposition \ref{th: time evolution} and the continuity of $\tau_s^{(a)}$ we can take the limit $m\rightarrow \infty$ and obtain $\alpha_t\circ\tau_s^{(a)}=\tau_s^{(a)}\circ\alpha_t$ on $\mathcal{U}_{\mathrm{loc}}$. The identity extends to all of $\mathcal{U}$ by density and continuity of the $*$-automorphisms. 
		\end{proof}
		
		\subsection{Checking the Mermin-Wagner criterion} \label{sect:relative entropy estimate}
		We are now in the position to spell out the core of the argument: we show that in the setting of Theorem~\ref{Thm:Main}, the criterion in Proposition \ref{th: Proposition 1} is satisfied. Here, as is typical for Mermin-Wagner type theorems, we make explicit use of the fact that the multipole symmetries are 1-parameter groups. 
		
		\begin{prop} \label{th: Proposition main Taylor construction}
			In the setting of Theorem~\ref{Thm:Main} with $k\in \mathbb{N}_0$, 
			if $a\in \mathbb{N}_0^d$ is a fixed multi-index with $\vert a\vert\leq k$ and $\gamma \leq 2(k-\vert a \vert +1)$, then~\eqref{eq: bound expectation value of U_m delta U_m^* not appendix} holds for $ U_m \equiv U_m^{(a)}(s) $ and any  $s\in\mathbb{R}$. \\[1ex]
			Assume additionally that there is a subset $I\subseteq\{1,...,d\}$ and constants $M_j\geq0$ such that $\vert x_{j}\vert\leq \vert M_j\vert$ for all $x\in\mathcal{L}$ and all $j\in I$. If \eqref{as:invariance} is weakened to 
            \begin{equation} \label{as: invariance weakened}
               \tau_s^{(a)}(\phi(\Lambda))=\phi(\Lambda)\text{ for all  }\Lambda \in P_0(\mathcal{L}), s\in\mathbb{R}, \mbox{and} \ a\in T_k  \, ,
            \end{equation}
            then \eqref{eq: bound expectation value of U_m delta U_m^* not appendix} still holds for multi-indices $a\in T_k \coloneqq \left\{ a\in \mathbb{N}_0^d \ \vert \ \vert a\vert\leq k \, \mbox{and $a_j=0$ for $j\in I$} \right\} $.
		\end{prop}

		\begin{proof}
			Since $\delta$ is the infinitesimal generator of $\alpha$, and $U_m^{(a)}(s)\in\mathcal{U}_{\mathrm{loc}}$, we have by \eqref{eq: time derivative},
			\begin{equation} \label{aux 1}
				\delta\bigl(U_m^{(a)}(s)\bigr)=i\sum_{\Lambda\in P_0(\mathcal{L})}\bigl[\phi(\Lambda), U_m^{(a)}(s)\bigr],
			\end{equation}
            where the sum has to be understood as a suitable limit as explained in \eqref{eq: time derivative}. Set 
			$$Q_m=\bigcup_{x\in\mathcal{L}\cap[-2m,2m]^d}\mathrm{supp}(n_x),$$ which contains the support of $U_m^{(a)}(s)$ for all $s\in \mathbb{R}$. Only terms with $\Lambda\cap Q_m\ne\emptyset$ yield a non-vanishing contribution to the sum in the right-hand side of~\eqref{aux 1}. 
			Thus, we estimate
            \begin{align} \label{eq:first step}
    \notag                    					&\Vert U_m^{(a)}(s)\delta\bigl(U_m^{(a)}(s)^*\bigr)+U_m^{(a)}(s)^*\delta\bigl(U_m^{(a)}(s)\bigr)\Vert \\
					\notag 
                    &=\Big\Vert \sum_{\Lambda\in P_0(\mathcal{L})}\bigl(U_m^{(a)}(s)\bigl[\phi(\Lambda), U_m^{(a)}(s)^*\bigr]+U_m^{(a)}(s)^*\bigl[\phi(\Lambda),U_m(s)\bigr]\bigr)\Big\Vert\\
					\notag 
                    &\leq\sum_{x\in Q_m}\sum_{\substack{\Lambda\in P_0(\mathcal{L}):\\x\in\Lambda}}\Vert U_m^{(a)}(s)\phi(\Lambda) U_m^{(a)}(s)^*-\phi(\Lambda) +U_m^{(a)}(s)^*\phi(\Lambda) U_m^{(a)}(s)-\phi(\Lambda) \Vert \\
					&\leq \vert Q_m\vert \sup_{x\in\mathcal{L}}\sum_{\substack{\Lambda\in P_0(\mathcal{L}):\\x\in\Lambda}}\Vert U_m^{(a)}(s)\phi(\Lambda) U_m^{(a)}(s)^*-\phi(\Lambda) +U_m^{(a)}(s)^*\phi(\Lambda) U_m^{(a)}(s)-\phi(\Lambda) \Vert.
            \end{align}
			We will now bound each terms with fixed $x\in\mathcal{L}$ and $\Lambda \in P_0(\mathcal{L})$ with $\Lambda\ni x$. 
             Since $\chi_m^{(a)}:z\mapsto z^a\chi_m(z)$ is in $C_c^\infty(\mathbb{R}^d)$, a Taylor expansion with Lagrange remainder around the point $x$ yields 
			\begin{equation}
				\label{q: expansion}
				\chi_m^{(a)}(z)=x^a\chi_m(x)+\sum_{l=1}^{k} \frac{1}{l!}D^l\chi_m^{(a)}(x)(z-x)+R_{k+1}(x,z),
			\end{equation}
			where 
			\begin{equation*}
				D^l\chi_m^{(a)}(x)(z-x)\coloneqq\sum_{\nu_1=1}^d...\sum_{\nu_l=1}^d \frac{\partial^l\chi_m^{(a)}}{\partial x_{\nu_1}...\partial x_{\nu_l}}(x) \prod_{j=1}^l (z_{\nu_j}-x_{\nu_j}).
			\end{equation*}
			There exists $\xi_{z,x}\in\{\lambda x+(1-\lambda) z\vert\; \lambda\in[0,1]\}$ such that 
			\begin{align}
					\notag 
                    \vert R_{k+1}(x,z)\vert&\leq\frac{1}{(k+1)!}\vert D^{k+1}\chi_m^{(a)}(\xi_{x,z})(z-x)\vert\\
					\notag &\leq \frac{1}{(k+1)!}\sum_{\nu_1=1}^d...\sum_{\nu_{k+1}=1}^d \vert \frac{\partial^{k+1}\chi_m^{(a)}}{\partial x_{\nu_1}...\partial x_{\nu_{k+1}}}(\xi_{z,x})\vert \prod_{j=1}^{k+1}\vert z_{\nu_j}-x_{\nu_j}\vert\\
					&\leq\frac{m^{\vert a\vert}}{m^{k+1}}\vert x-z\vert^{k+1} \frac{1}{(k+1)!}\sup_{\xi\in\mathbb{R}^d}\sum_{\nu_1=1}^d...\sum_{\nu_{k+1}=1}^d\Bigl\vert\frac{\partial^{k+1}\chi^{(a)}}{\partial x_{\nu_1}...\partial x_{\nu_{k+1}}}(\xi)\Bigr\vert,
				\end{align}
			where we have used that by definition $\chi_m^{(a)}(z)=z^a\chi(z/m)=m^{\vert a\vert}(z/m)^a\chi(z/m)=m^{\vert a\vert}\chi^{(a)}(z/m)$.  
			Note that 
            \begin{equation} \label{def Cb}
                C_a=\frac{1}{(k+1)!}\sup_{\xi\in\mathbb{R}^d}\sum_{\nu_1=1}^d...\sum_{\nu_{k+1}=1}^d\Bigl\vert\frac{\partial^{k+1}\chi^{(a)}}{\partial x_{\nu_1}...\partial x_{\nu_{k+1}}}(\xi)\Bigr\vert
            \end{equation}
			is independent of $m$ and finite since $\chi^{(a)}\in C_c^\infty (\mathbb{R})$. Thus,
			\begin{equation}
				\label{eq: taylor for f_m}
				\vert R_{k+1}(x,z)\vert\leq \frac{C_a}{m^{k-\vert a\vert +1}} \vert x-z\vert^{k+1}.
			\end{equation}
			We now substitute the expansion \eqref{q: expansion} for $\chi_m^{(a)}(z)$ into the definition of $U_m^{(a)}(s)$ and regroup the terms according to their order in $z$ to obtain 
			\begin{align*}
				U_m^{(a)}(s)& =\exp{\Bigl(is\sum_{\substack{z\in\mathcal{L}:\\\mathrm{supp}(n_z)\cap\Lambda\ne\emptyset}}\chi_m^{(a)}(z)n_z\Bigr)}\\
    & =\exp{\Bigl(is\sum_{\substack{z\in\mathcal{L}:\\\mathrm{supp}(n_z)\cap\Lambda\ne\emptyset}}R_{k+1}(x,z)n_z\Bigr)} \exp{\Bigl(\sum_{l=0}^{k}\sum_{\substack{\nu_1,...,\nu_d:\\\nu_1+...+\nu_d=l}}\bigl(isa_{\nu_1,...,\nu_d}(x)\sum_{\substack{z\in\mathcal{L}:\\\mathrm{supp}(n_z)\cap\Lambda\ne\emptyset}}z_{\nu_1}...z_{\nu_d} \, n_z\bigr)\Bigr)}
			\end{align*}
			with coefficients $a_{\nu_1,...,\nu_d}(x)$ depending on $\chi_m$ and $x$. 
			Recall that by Assumption \ref{as: family of charge operators} the $n_z$ pairwise commute. Therefore, the explicit form \eqref{eq:symmetries}  of the symmetry automorphism groups and the symmetry assumption \eqref{as:invariance} allow to systematically cancel all the exponential terms except for the one with the remainder $R_{k+1}$ in the expression for $U_m^{(a)}(s)\phi(\Lambda) U_m^{(a)}(s)^*$. This leaves us with
			\begin{align}
				\label{eq: second step}
				U_m^{(a)}(s)\phi(\Lambda) U_m^{(a)}(s)^* 
    =\exp{\Bigl(is\sum_{\substack{z\in\mathcal{L}:\\\mathrm{supp}(n_z)\cap\Lambda\ne\emptyset}}R_{k+1}(x,z)n_z\Bigr)}\phi(\Lambda)\exp{\Bigl(-is\sum_{\substack{z\in\mathcal{L}:\\\mathrm{supp}(n_z)\cap\Lambda\ne\emptyset}}R_{k+1}(x,z)n_z\Bigr)}.
				\end{align}
			We now define $F_{\Lambda,x}:\mathbb{R}\rightarrow\mathcal{U}$ by
				\begin{align}
				F_{\Lambda,x}(r) \coloneqq &\exp\Bigl({-ir\sum_{\substack{z\in\mathcal{L}:\\\mathrm{supp}(n_z)\cap\Lambda\ne\emptyset}}R_{k+1}(x,z)n_z}\Bigr)\phi(\Lambda) \exp\Bigl({ir\sum_{\substack{z\in\mathcal{L}:\\\mathrm{supp}(n_z)\cap\Lambda\ne\emptyset}}R_{k+1}(x,z)n_z}\Bigr)-\phi(\Lambda) \notag \\
					& + \exp\Bigl({ir\sum_{\substack{z\in\mathcal{L}:\\\mathrm{supp}(n_z)\cap\Lambda\ne\emptyset}}R_{k+1}(x,z)n_z}\Bigr)\phi(\Lambda) \exp\Bigl({-ir\sum_{\substack{z\in\mathcal{L}:\\\mathrm{supp}(n_z)\cap\Lambda\ne\emptyset}}R_{k+1}(x,z)n_z}\Bigr)-\phi(\Lambda). 
			\end{align}
			Combining \eqref{eq: second step} and an analogous expression for $U_m^{(a)}(s)^*\phi(\Lambda) U_m^{(a)}(s)$ re-expresses the term to be bounded as  
			\begin{equation*}
				U_m^{(a)}(s)\phi(\Lambda) U_m^{(a)}(s)^*-\phi(\Lambda) +U_m^{(a)}(s)^*\phi(\Lambda) U_m^{(a)}(s)-\phi(\Lambda)=F_{\Lambda,x}(s).
			\end{equation*}
			Note that $F_{\Lambda,x}$ is a $\mathcal{U}$-valued $C^2$-function, and thus by Taylor's theorem in Banach spaces:
			$$F_{\Lambda,x}(s)=F_{\Lambda,x}(0)+F'_{\Lambda,x}(0)\cdot s+\mathfrak{R}_2(s), \quad \mbox{with} \quad \Vert \mathfrak{R}_2(s)\Vert \leq \frac{s^2}{2} \Vert F_{\Lambda,x}''(s_0)\Vert  $$ 
			at some $ s_0\in[0,s]$. 
			Clearly $F_{\Lambda,x}(0)=0$ and 
			\begin{equation*}
				F'_{\Lambda,x}(0)=-i\Bigl[\sum_{\substack{z\in\mathcal{L}:\\\mathrm{supp}(n_z)\cap\Lambda\ne\emptyset}}R_{k+1}(x,z)n_z, \phi(\Lambda)\Bigr]+i\Bigl[\sum_{\substack{z\in\mathcal{L}:\\\mathrm{supp}(n_z)\cap\Lambda\ne\emptyset}}R_{k+1}(x,z)n_z, \phi(\Lambda)\Bigr]=0.
			\end{equation*}
			Furthermore, we have
			\begin{align*}
				F''_{\Lambda,x}(s_0)=
				 -\Biggl[&\sum_{\substack{z\in\mathcal{L}:\\\mathrm{supp}(n_z)\cap\Lambda\ne\emptyset}}  R_{k+1}(x,z) n_z,\\
				&\qquad \Bigl[\sum_{\substack{y\in\mathcal{L}:\\\mathrm{supp}(n_y)\cap\Lambda\ne\emptyset}}R_{k+1}(x,y)n_y,U_m^{(a)}(s_0)^*\phi(\Lambda) U_m^{(a)}(s_0)+U_m^{(a)}(s_0)\phi(\Lambda) U_m^{(a)}(s_0)^*\Bigr]\Biggr].
			\end{align*}
			Thus, with \eqref{eq: taylor for f_m} and $\Vert n_z\Vert\leq N_0$ for all $z$, we arrive at the following estimate with $C_a$ from \eqref{def Cb}:
			\begin{align*}
					\Vert F_{\Lambda,x}(s)\Vert
                    &\leq\frac{s^2}{2}\Vert F''_{\Lambda,x}(s_0)\Vert \\
					&\leq \frac{s^2}{2} \sum_{\substack{z,y\in\mathcal{L}:\\\mathrm{supp}(n_z)\cap\Lambda\ne\emptyset\\\mathrm{supp}(n_y)\cap\Lambda\ne\emptyset}}\vert R_{k+1}(x,z) R_{k+1}(x,y)\vert \\
					&\qquad \qquad \qquad \qquad \qquad \times \left\Vert\Bigl[n_z,\bigl[n_y,U_m^{(a)}(s_0)^*\phi(\Lambda) U_m^{(a)}(s_0)+U_m^{(a)}(s_0)\phi(\Lambda) U_m^{(a)}(s_0)^*\bigr]\Bigr]\right\Vert \\
					&\leq 4s^2 C_a^2 N_0^2 \frac{1}{m^{2(k-\vert a\vert+1)}}\sum_{\substack{z\in\mathcal{L}:\\\mathrm{supp}(n_z)\cap\Lambda\ne\emptyset}}\sum_{\substack{y\in\mathcal{L}:\\\mathrm{supp}(n_y)\cap\Lambda\ne\emptyset}}\vert x-z\vert^{k+1}\vert x-y\vert^{k+1}\Vert \phi(\Lambda)\Vert \\
					&=4s^2 C_a^2 N_0^2 \frac{1}{m^{2(k-\vert a\vert+1)}}\Bigl(\sum_{\substack{z\in\mathcal{L}:\\\mathrm{supp}(n_z)\cap\Lambda\ne\emptyset}}\vert x-z\vert^{k+1}\Bigr)^2\Vert \phi(\Lambda)\Vert.
				\end{align*}
			Lastly, combining this estimate with \eqref{eq:first step} yields
			\begin{align}
				\label{eq: last step}
					&\Vert U_m^{(a)}(s)\delta(U_m^{(a)}(s)^*)+U_m^{(a)}(s)^*\delta(U_m^{(a)}(s))\Vert \notag\\
					&\leq 4s^2 C_a^2 N_0^2 \frac{\vert Q_m\vert}{m^{2(k-\vert a\vert+1)}} \sup_{x\in \mathcal{L}} \Biggl(\sum_{\substack{\Lambda\in P_0(\mathcal{L}):\\x\in\Lambda}}\Vert \phi(\Lambda)\Vert \Bigl(\sum_{\substack{z\in\mathcal{L}:\\\mathrm{supp}(n_z)\cap\Lambda\ne\emptyset}}\vert x-z\vert^{k+1}\Bigr)^2\Biggr).
				\end{align}
			The supremum is finite by the assumption \eqref{decay condition}. By Assumption \eqref{as: family of charge operators} we have $\mathrm{supp}(n_z)\subseteq B_{R_0}(z)$ for all $z\in \mathcal{L}$, and therefore, we bound
			\begin{equation}
				\vert Q_m\vert=\Big\vert\bigcup_{z\in\mathcal{L}\cap[-2m,2m]^d}\mathrm{supp}(n_z)\Big\vert\leq \vert\mathcal{L}\cap[-2m-R_0,2m+R_0]^d\big\vert .
			\end{equation}
			Since $\vert\mathcal{L}\cap[-m,m]^d\vert=\mathcal{O}(m^\gamma) = \mathcal{O}(m^{2(k-\vert a\vert+1)})$ by Assumption~\eqref{as: lattice}, the right side in \eqref{eq: last step} is the desired $m$-independent bound. This finishes the first part of the proof. 
			
			For a proof of the second part, we fix $I\subseteq\{1,...,d\}$ as stated. 
			By construction, $x_j\mapsto \chi_m(x_1 ,...., x_j,..., x_d)$ is constant on $[-m, m]$. 
			Therefore, for $m>M :=\max_{j\in I}M_j$ and for multi-indices $a$ with $a_j=0$ for all $j\in I$, we obtain that for every fixed $j\in I$, the map $x_j\mapsto x^a\chi_m(x_1,...,x_j,...,x_d)$ is constant in an open neighborhood of $[-M,M]$. This in particular implies 
			\begin{equation}
				\label{eq: certain partial derivatives vanish}
				\partial_j \chi_m^{(a)}(x)=0 \;\; \mathrm{for}\;\mathrm{all}\; x\in\mathcal{L},\; j\in I \;\mathrm{and}\; m>M.
			\end{equation}
			We now proceed as before. 
			Taylor expanding $z\mapsto \chi^{(a)}_m(z)$ around the point $x$ as in \eqref{q: expansion} yields 
			\begin{equation*}
				\chi_m^{(a)}(z)=\sum_{l=0}^{k}\sum_{\substack{\nu_1,...,\nu_d:\\\nu_1+...+\nu_d=l}}\bigl(isa_{\nu_1,...,\nu_d}(x)z_{\nu_1}...z_{\nu_d}\bigr)+R_{k+1}(x,z),
			\end{equation*}
			where the coefficients $a_{\nu_1,...,\nu_d}(x)$ again depend on $x$ and $\chi_m$. However, now we have $a_{\nu_1,...,\nu_d}=0$ if there is a $j\in I$ with $\nu_j\ne0$ due to \eqref{eq: certain partial derivatives vanish}. 
			We then plug this into the expression for $U_m^{(a)}(s)\phi(\Lambda)U_m^{(a)}(s)^*$ and use the symmetry assumption to eliminate all terms in the expansion except for the remainder and end up with \eqref{eq: second step}. For this, we only need the multipole symmetries associated with multi-indices $a$ with $a_j=0$ for all $j\in I$; that is \eqref{as: invariance weakened} is already sufficient, and we do not require \eqref{as:invariance}. We then finish the proof analogously as for the first part.  
		\end{proof}
		
		We can now combine Proposition \ref{th: Proposition 1} and \ref{th: Proposition main Taylor construction} to prove Theorem \ref{Thm:Main}.
		\begin{proof}[Proof of Theorem \ref{Thm:Main}]
		    We apply Proposition \ref{th: Proposition 1} with $U_m:=U_m^{(a)}(s)$ and $\Lambda_m=[-m+R_0,m-R_0]^d\cap\mathcal{L}$.
		Clearly $\Lambda_m\uparrow \mathcal{L}$ and, by construction, we have $\chi_m(x)=1$ if $| x | _\infty\leq m$. 
		Thus, if $x\in\mathcal{L}$ is such that $\mathrm{supp}(n_x)\cap \Lambda_m\ne\emptyset$, then $| x|_\infty\leq m-R_0+R_0=m$. Hence $\chi_m(x)=1$ and  Proposition~\ref{th: definition of the multipole symmetries} ensures the validity of\eqref{eq:symmetries}, i.e., $\tau_s^{(a)}(A)=U_m^{(a)}(s)^*AU_m^{(a)}(s)$  for all  $A\in\mathcal{U}_{\Lambda_m}$. 
		Hence, Proposition \ref{th: Proposition main Taylor construction} establishes that Proposition \ref{th: Proposition 1} is applicable and Theorem \ref{Thm:Main} follows. 
		\end{proof}

		\begin{rem} \label{rem: no ssb for beta=0}
			The multipole symmetries of Proposition \ref{th: definition of the multipole symmetries} can never be spontaneously broken on infinite\--temperature (that is $\beta=0$) KMS states of any time-evolution $\alpha$ on $\mathcal{U}$. Let $\omega\in K_0(\alpha)$ be an infinite temperature KMS state and $a\in\mathbb{N}_0^d$ some multi-index. Then we have, by definition, $\omega(AB)=\omega(A\alpha_0(B))=\omega(BA)$ for all $A,B\in \mathcal{U}$. This tracial property of $\omega$ immediately implies that automatically  $\omega\circ\tau^{(a)}_s=\omega$ for all $s\in\mathbb{R}$. Indeed, we have for $A\in\mathcal{U}$
			\begin{equation}
				\omega\circ\tau_s^{(a)}(A)=\lim_{m\rightarrow\infty}\omega\bigl(U_m^{(a)}(s)^*AU_m^{(a)}(s)\bigr)=\lim_{m\rightarrow\infty}\omega\bigl(AU_m^{(a)}(s)U_m^{(a)}(s)^*\bigr)=\omega(A).
			\end{equation}
			In particular, no assumptions on the spatial dimension $d$ or the decay of the interactions are needed. 
		\end{rem}
		
		\appendix

		\section{Time-evolution in infinite volume} \label{sect:Dynamics}
		
		The goal of this appendix is twofold. We first compile basic results that guarantee the existence of the time-evolution in infinite volume induced by an interaction (see \cite{nachtergaele2006propagation,nachtergaele2019quasi} and references therein, and \cite{LectureNotes,naaijkens2017quantum} for a pedagogical introduction). We also provide sufficient conditions implying that the assumptions of our main theorem are satisfied. 
		
		\subsection{General theory}
Following~\cite{nachtergaele2006propagation,nachtergaele2019quasi,LectureNotes}, we find it convenient to characterize the decay conditions on interactions, which ensure the existence of the time-evolution on the algebra $ \mathcal{U} $ of quasi-local observables in terms of $ F $-functions. 
    \begin{defn} ($F$-functions) 
			\label{def: F functions}
			A map $F:[0, \infty)\rightarrow (0, \infty)$ will be called an $F$-function for the countable metric space $(\mathcal{L},D)$ if it satisfies 
			\begin{enumerate}
				\item[$(a)$] $F(r)\geq F(s)$ for all $0\leq r\leq s$.
				\item[$(b)$] $\Vert F \Vert\coloneqq\sup_{x\in\mathcal{L}}\sum_{y\in\mathcal{L}}F\bigl(D(x,y)\bigr)<\infty$.
				\item[$(c)$] There exists $C_F\in[0,\infty)$ such that for all $x,y\in\mathcal{L}$
				\begin{equation} \label{convolution property}
					\sum_{z\in\mathcal{L}}F\bigl(D(x,z)\bigr)F\bigl(D(z,y)\bigr)\leq C_F F\bigl(D(x,y)\bigr).
				\end{equation}
			\end{enumerate}
            \end{defn}
            In case $\mathcal{L}=\mathbb{Z}^d$ equipped with the Euclidean metric, it is well known~\cite[Sec. 1.1]{nachtergaele2006propagation} that $F:[0,\infty)\ni r\mapsto (1+r)^{-d-\epsilon}$ is an $F$-function for every $\epsilon >0$. 
            For more general metric spaces with effective dimension $ \gamma $ considered in Assumption~\ref{as: lattice}, we have the following
            \begin{lemma} \label{lem: suff. cond for decay cond.}
			In the set-up of Assumption~\ref{as: lattice}  if $\lambda>1+\gamma$, then $F:[0,\infty)\rightarrow \mathbb{R}, r\mapsto (1+r)^{-\lambda}$ is an $F$-function for $(\mathcal{L},D)$.
		\end{lemma}
		\begin{proof}
			If $\nu>0$, then for $x\in\mathcal{L}$ fixed and with $ r_0 > 0 $ and the effective dimension $ \gamma $ from \eqref{as: lattice}:
			\begin{align}
				\label{eq: lemma 2 summability estimate}
				\sum_{z\in\mathcal{L}}\bigl(1+\vert x-z\vert\bigr)^{-\nu}
					&\leq \sum_{j=0}^\infty\big\vert\mathcal{L}\cap B_{(j+1)r_0}(x)\big\vert\bigl(1+jr_0\bigr)^{-\nu}\leq \sum_{j=0}^\infty C\frac{\bigl((j+1)r_0\bigr)^\gamma}{\bigl(1+jr_0\bigr)^\nu}.
				\end{align}
			The right side is finite and independent of $x$ if $\nu>\gamma+1$. Furthermore, $F$ is clearly nonnegative and decreasing. The convolution property \eqref{convolution property} is checked by elementary arguments.
            \end{proof}
        Given an $F$-function, the vector space 
			\begin{equation}\label{def: BF}
				\mathcal{B}_F(\mathcal{L})\coloneqq\Bigl\{\Phi\Big\vert\;\Phi\;\mathrm{is}\;\mathrm{an}\;\mathrm{interaction}\;\mathrm{and}\;\sup_{x,y\in\mathcal{L}}\frac{1}{F(D(x,y))}\sum_{\Lambda\in P_0(\mathcal{L}): \;x,y\in \Lambda}\Vert\phi(\Lambda)\Vert<\infty\Bigr\}
			\end{equation}
            of interactions $ \Phi = ( \phi(\Lambda))_{\Lambda \in P_0(\mathcal{L}} $ is known \cite[Sec.~4]{LectureNotes} to be a Banach space when equipped with the 
            norm 
			\begin{equation}
				\label{eq: F norm definition}
				\Vert \Phi \Vert_F=\sup_{x,y\in\mathcal{L}}\frac{1}{F(D(x,y))}\sum_{\Lambda\in P_0(\mathcal{L}): \;x,y\in \Lambda}\Vert\phi(\Lambda)\Vert . 
			\end{equation} 
			The sum on the right side is to be understood in the sense of $\lim_{n\rightarrow\infty}\sum_{\Lambda\subseteq\Lambda_n: \;x,y\in \Lambda}\Vert\phi(\Lambda)\Vert$ with any sequence $\Lambda_n\uparrow\mathcal{L}$ and the limit is independent of the choice of the finite, increasing and exhausting $(\Lambda_n)_n$. In fact, it is easy to see that for any $ \Phi \in \mathcal{B}_F(\mathcal{L}) $ and $X\in P_0(\mathcal{L})$ the limit
				\begin{equation}
					\label{eq: sufrace energy limit}
					\sum_{\Lambda\in P_0(\mathcal{L}), \Lambda\cap X\ne\emptyset}\phi(\Lambda) \coloneqq \lim_{n\rightarrow\infty}\sum_{\substack{\Lambda\subseteq\Lambda_n:\\\Lambda\cap X\ne\emptyset}}\phi(\Lambda)
				\end{equation}
				exists in $\mathcal{U}$ and is independent of the sequence $\Lambda_n \uparrow \mathcal{L} $. \\

				The following is the content of \cite[Thm. 5.3]{LectureNotes} (see also \cite{nachtergaele2006propagation}) and \cite[Prop. 6.2.4]{BratelliRobinson2} reformulated for our set-up. 
				\begin{prop} 
			\label{th: time evolution}
			Let $\mathcal{U}$ be the $C^*$-algebra of quasi-local operators associated with $(\mathcal{L},D)$,  $F$ an $F$-function as in Definition \ref{def: F functions} and  $\Phi\in \mathcal{B}_F(\mathcal{L})$.
			\begin{enumerate}
				\item[(i)] For any $t\in\mathbb{R}$ and $A\in\mathcal{U}$ and any finite, increasing and exhausting sequence $\Lambda_n\uparrow\mathcal{L}$, the limit
				\begin{equation}
					\label{eq: interaction limit}
					\alpha_t(A) \coloneqq \lim_{n\rightarrow \infty}\exp{\Bigl(it\sum_{\Lambda\subseteq\Lambda_n}\phi(\Lambda)\Bigr)}A\exp{\Bigl(-it\sum_{\Lambda\subseteq\Lambda_n}\phi(\Lambda)\Bigr)}
				\end{equation}
                exists in $ \mathcal{U} $, is independent of the sequence $(\Lambda_n)_n$, and the convergence is uniform in $ t $ on compact subsets.
				The limit defines a strongly continuous 1-parameter group $(\alpha_t)_{t\in\mathbb{R}}$ of $*$-automorphisms of $\mathcal{U}$. 
				\item[(ii)] The infinitesimal generator $\delta$  of $\alpha$ includes the local algebra in its domain, $\mathcal{U}_{\mathrm{loc}}\subseteq D(\delta)$.  For any $A\in\mathcal{U}_{\mathrm{loc}}$ and any sequence $ \Lambda_n \uparrow \mathcal{L} $
				\begin{equation}
					\label{eq: time derivative}
					\delta(A)= i\sum_{\Lambda\in P_0(\mathcal{L})}[\phi(\Lambda),A] \coloneqq i\lim_{n\rightarrow \infty}\Bigl[\sum_{\Lambda\subseteq\Lambda_n}\phi(\Lambda),A\Bigr] . 
				\end{equation}
			\end{enumerate}
		\end{prop}
	
		\subsection{Sufficient conditions}
		
		In this subsection we establish some sufficient conditions for the assumptions of our main theorem to hold. We start by spelling out a sufficient condition on the $F$-function to satisfy the decay condition~\eqref{decay condition} in Assumption~\ref{as: interaction}.
		
		\begin{lemma}
			\label{th: lemma sufficient condition for (i)}
			Consider a countable metric space $(\mathcal{L}, D)$ that satisfies Assumption \ref{as: lattice} and a family of charge operators $\{n_z\vert\;z\in\mathcal{L}\}$ as in Assumption \ref{as: family of charge operators}. Given $\Phi\in \mathcal{B}_F(\mathcal{L})$ with an $F$-function that satisfies
			\begin{equation}
				\label{eq: stronger integrability contion}
				\sup_{x\in\mathcal{L}}\sum_{z\in\mathcal{L}}\Bigl(\sum_{z'\in B_{R_0}(z)\cap\mathcal{L}}F\bigl(\vert x-z'\vert\bigr)\Bigr)^{\frac{1}{2}}\vert x-z\vert^{k+1}<\infty
			\end{equation}
			with the constant $R_0$ from Assumption \ref{as: family of charge operators}, then~\eqref{decay condition} holds. 
		\end{lemma}
		
		\begin{proof}
			For $\Lambda\in P_0(\mathcal{L})$ we abbreviate $E(\Lambda)\coloneqq\{z\in\mathcal{L} \ \vert \ \mathrm{supp}(n_z)\cap\Lambda\ne\emptyset\}$. 
			If  $\Lambda_n\uparrow\mathcal{L}$ and $x\in\mathcal{L}$ is fixed, we conclude from~\eqref{eq: sufrace energy limit}:
			\begin{align}
				\label{eq: appendix computation}
					\notag\sum_{\substack{\Lambda\in P_0(\mathcal{L}):\\x\in \Lambda}}\Vert\phi(\Lambda)\Vert&\Bigl(\sum_{\substack{z\in\mathcal{L}:\\\mathrm{supp}(n_z)\cap\Lambda\ne\emptyset}}\vert z-x\vert^{k+1}\Bigr)^2=\lim_{n\rightarrow\infty}\sum_{\substack{\Lambda\subseteq\Lambda_n:\\x\in \Lambda}}\Vert\phi(\Lambda)\Vert\sum_{y,z\in E(\Lambda)}\vert x-y\vert^{k+1}\vert x-z\vert^{k+1}\\
					\notag &=\lim_{n\rightarrow\infty}\sum_{y,z\in\mathcal{L}}\sum_{\substack{\Lambda\subseteq\Lambda_n:\\x\in\Lambda, y,z\in E(\Lambda)}}\Vert\phi(\Lambda)\Vert\vert x-y\vert^{k+1}\vert x-z\vert^{k+1}\\&\leq\lim_{n\rightarrow\infty}\sum_{y,z\in\mathcal{L}}\Bigl(\sum_{\substack{\Lambda\subseteq\Lambda_n:\\x\in\Lambda, y\in E(\Lambda)}}\Vert\phi(\Lambda)\Vert\Bigr)^{\frac{1}{2}}\vert x-y\vert^{k+1}\Bigl(\sum_{\substack{\Lambda\subseteq\Lambda_n:\\x\in\Lambda, z\in E(\Lambda)}}\Vert\phi(\Lambda)\Vert\Bigr)^{\frac{1}{2}}\vert x-z\vert^{k+1}.
				\end{align}
			By Definition \ref{as: family of charge operators} we have $\mathrm{supp}(n_z)\subseteq B_{R_0}(z)$ . Hence, for fixed $x,z\in\mathcal{L}$ by definition of $E(\Lambda)$:
			\begin{equation}
				\sum_{\substack{\Lambda\subseteq\Lambda_n:\\x\in\Lambda, z\in E(\Lambda)}}\Vert\phi(\Lambda)\Vert\leq\sum_{z'\in B_{R_0}(z)\cap\mathcal{L}} \, \sum_{\substack{\Lambda\subseteq\Lambda_n:\\x,z'\in \Lambda}}\Vert\phi(\Lambda)\Vert\leq \Vert\Phi\Vert_F\sum_{z'\in B_{R_0}(z)\cap\mathcal{L}} F(\vert x-z'\vert),
			\end{equation}
			where we used the definition~\eqref{eq: F norm definition} of the $F$-norm. 
			Combining this with \eqref{eq: appendix computation} yields 
			\begin{equation*}
					\sum_{\substack{\Lambda\in P_0(\mathcal{L}):\\x\in \Lambda}}   \Vert \phi(\Lambda)\Vert \Bigl(\sum_{\substack{z\in\mathcal{L}:\\\mathrm{supp}(n_z)\cap\Lambda\ne\emptyset}}\vert z-x\vert^{k+1}\Bigr)^2\\
					\leq 
                    \Vert\Phi\Vert_F\Biggl(\sum_{z\in\mathcal{L}}\Bigl(\sum_{z'\in B_{R_0}(z)\cap\mathcal{L}}F\bigl(\vert x-z'\vert\bigr)\Bigr)^{\frac{1}{2}}\vert x-z\vert^{k+1}\Biggr)^2.
			\end{equation*}
			Hence \eqref{eq: stronger integrability contion} guarantees~\eqref{decay condition}.  
		\end{proof}
		As a corollary, we show that fast enough polynomial decay is sufficient for the growth condition \eqref{decay condition} to hold true.
		\begin{cor} \label{cor: suff. cond for decay cond.}
			Suppose that Assumption \ref{as: lattice} is satisfied.
			 If $\lambda>4+2k+2\gamma$ and $F:[0,\infty)\rightarrow \mathbb{R}, r\mapsto (1+r)^{-\lambda}$ is an $F$-function for $(\mathcal{L}, D)$ with $\Phi \in \mathcal{B}_F(\mathcal{L}) $, then $\Phi$ satisfies the decay condition \eqref{decay condition}.
		\end{cor}
		\begin{proof} By Lemma \ref{th: lemma sufficient condition for (i)} it is enough to show that
			\begin{equation}
				\sup_{x\in\mathcal{L}}\sum_{z\in\mathcal{L}}\Bigl(\sum_{z'\in B_{R_0}(z)\cap\mathcal{L}}\bigl(1+\vert x-z'\vert\bigr)^{-\lambda}\Bigr)^{\frac{1}{2}}\vert x-z\vert^{k+1}<\infty,
			\end{equation}
			where $R_0$ is the constraint on the support of the $n_z$ as defined in Assumption \ref{as: family of charge operators}. This is a straightforward computation using \eqref{as: lattice} and \eqref{eq: lemma 2 summability estimate}.
		\end{proof}
		
\section{Mermin-Wagner criterion} \label{sect:PfisterFröhlich}

    The Mermin-Wagner criterion in Proposition~\ref{th: Proposition 1} is known, yet somwhat hidden in the literature. It appeared in Fr\"ohlich-Pfister's proof \cite{PfisterPaper} of the absence of symmetry breaking for two-body translation-invariant potentials on $\mathbb{Z}^2$. A related entropy-based argument is in \cite{FannesEntropySSB,NachtergaeleMerminWagner} (see also \cite[Thm. 5.3.33]{BratelliRobinson2}). 
    For the convenience of the reader, we sketch the argument in \cite[Ch.~3]{PfisterPaper}. \\

  To prove that
		$	\omega \circ \tau = \omega $
		for all KMS-states $\omega \in K_\beta(\alpha)$ and the multipole symmetry $ \tau $, the key quantity to control is the relative entropy, which in our set-up is the monotone limit  of the state's restriction to the finite-dimensional subalgebras $ \mathcal{U}_{\Lambda_m} $ for any $ \Lambda_m \uparrow \mathcal{L} $:
        \begin{equation} \label{relative entropy as limit}
           0\leq  S(\omega \vert \omega\circ \tau) = \lim_{m\rightarrow\infty} S\left( (\omega \vert_{\mathcal{U}_{\Lambda_m}}) \vert ((\omega \circ \tau) \vert_{\mathcal{U}_{\Lambda_m}}) \right) = \sup_{m \in \mathbb{N}} S\left( (\omega \vert_{\mathcal{U}_{\Lambda_m}}) \vert ((\omega \circ \tau) \vert_{\mathcal{U}_{\Lambda_m}}) \right) ,
        \end{equation}
        see~\cite[Prop. 5.23]{OhyaPetzQuantumEntropy}. 
        The fact that the limit is monotone is known as Uhlman's monotonicity \cite[Thm. 5.3]{OhyaPetzQuantumEntropy}. A key fact is the following criterion.
        \begin{lemma}\label{lem:fact1}
            If $ S(\omega \vert \omega \circ \tau)<\infty$ for all $\omega \in K_\beta(\alpha) $, then $ \omega = \omega \circ \tau$ for all $\omega \in K_\beta(\alpha) $.
        \end{lemma}
        \begin{proof}
            By \cite[Thm.~5.3.30]{BratelliRobinson2} the finiteness of the relative entropy entails that $\omega=\omega\circ \tau$ for all extremal KMS-states $ \omega $. Note that symmetries of $\alpha$ leave $K_\beta(\alpha)$ invariant \cite[Prop.~5.3.33]{BratelliRobinson2} and therefore $\omega\circ \tau \in K_\beta(\alpha) $ is again extremal.
                As the set of KMS-states is convex \cite[Thm.~5.3.30]{BratelliRobinson2}, this implies that precomposing by $\tau$ leaves any KMS-states invariant. 
        \end{proof}

        In the setting of Proposition~\ref{th: Proposition 1}, the symmetry $ \tau $ is approximated by a sequence $ \tau_m : \mathcal{U} \ni A \mapsto U_m^* A U_m $ generated by local unitaries $(U_m)_{m\in\mathbb{N}}\subseteq\mathcal{U}_{\mathrm{loc}}$ corresponding to $ \Lambda_m \uparrow \mathcal{L} $ such that for all $m\in\mathbb{N}$
			\begin{equation} \label{local unitaries}
				\tau(A)=\tau_m(A) \;\mathrm{for}\;\mathrm{all}\; A\in\mathcal{U}_{\Lambda_m}.
			\end{equation}
        This implies that the right side in~\eqref{relative entropy as limit} is bounded according to
			\begin{equation}\label{relative entropy bdd}
		S\Bigl((\omega\vert_{\mathcal{U}_{\Lambda_m}})\Big\vert \bigl((\omega\circ\tau)\vert_{\mathcal{U}_{\Lambda_m}}\bigr)\Bigr)=S\Bigl((\omega\vert_{\mathcal{U}_{\Lambda_m}})\Big\vert \bigl((\omega\circ\tau_m)\vert_{\mathcal{U}_{\Lambda_m}}\bigr)\Bigr)\leq S(\omega\vert \omega\circ\tau_m),
			\end{equation}
			where the inequality is Uhlman's monotonicity. 
            In order to further estimate $ S(\omega\vert \omega\circ\tau_m) $, one represents  $ \omega \circ \tau_m   $ as a perturbation of the original $ \omega \in K_\beta(\alpha)$, see Lemma~\ref{lem:perturbo} below.\\ 

        To briefly review the notion of perturbation of a KMS-state, let $(\mathcal{H}, \pi, \Omega)$ be the GNS-triple associated to $ \omega $ on $ \mathcal{U}$, i.e., $\mathcal{H}$ is a Hilbert space, $\pi: \mathcal{U}\rightarrow L(\mathcal{H})$ is a multiplicative and additive map from the quasi-local algebra to the bounded operators on $\mathcal{H}$ such that for all $A\in \mathcal{U}$
        \begin{equation}
            \pi(A^*)=\pi(A)^*, \qquad \omega(A)=\langle \Omega, \pi(A)\Omega\rangle. 
        \end{equation}
        (For an overview of the GNS-construction, see \cite[Sec. 7.1]{LectureNotes} and \cite[Ch. 2.5]{naaijkens2017quantum}.) Given the time-evolution~$ \alpha $ on $ \mathcal{U}$,   there exists a self-adjoint operator,  $H:\mathcal{H}\supseteq D(H)\rightarrow\mathcal{H}$, called the GNS-Hamiltonian. The unitary group $(e^{itH})_{t\in\mathbb{R}}$ implements the time evolution $\alpha$ in the GNS-representation
		\begin{equation}
		\pi\bigl(\alpha_t(A)\bigr)=e^{itH}\pi(A)e^{-itH}\;\;\mathrm{for}\;\mathrm{all}\;t\in\mathbb{R},\;A\in\mathcal{U}.
		\end{equation}
		For a self-adjoint element $W\in \mathcal{U}$, one then shows \cite[Thm. 5.4.4]{BratelliRobinson2} that $\Omega\in D\bigl(e^{-\frac{\beta}{2}(H-\pi(W))}\bigr)$ and
				\begin{equation} \label{pertubed dynamics}
					\omega^W:\mathcal{U}\ni A\mapsto \bigl\langle e^{-\frac{\beta}{2}(H-\pi(W))}\Omega,\pi(A)e^{-\frac{\beta}{2}(H-\pi(W))}\Omega\bigr\rangle
				\end{equation}
				defines non-negative linear functional on $\mathcal{U}$ (not necessarily normalized). Furthermore, $\omega^W$ satisfies the $\beta$-KMS condition for a perturbed time-evolution $\alpha^W$ on $ \mathcal{U} $, which is constructed via a Dyson series \cite[Thm. 5.4.4]{BratelliRobinson2} and satisfies 
                \begin{equation}
                \pi(\alpha_t^W(A)) = e^{it(H-\pi(W))}\pi(A)e^{-it(H-\pi(W))} \;\;\mathrm{for}\;\mathrm{all}\;t\in\mathbb{R},\;A\in\mathcal{U}.
                \end{equation}
                Finally, we recall from \cite[Prop. 6.2.32]{BratelliRobinson2}   that 
            \begin{equation} \label{entropy to state}
                S(\omega\vert\omega^W) = - \beta\omega(W) 
            \end{equation}
            for all self-adjoint $ W \in \mathcal{U}$.\\

        We apply the above construction with two choices:\\

        \noindent
       1.~The original KMS-state $ \omega \in K_\beta(\alpha) $ and $ W $ equal to 
        \begin{equation}\label{delta to hm}
            h_m \coloneqq  \sum_{\substack{\Lambda \in P_0(\mathcal{L}) : \\ \Lambda \cap Q_m \neq 0}} \phi(\Lambda) ,
        \end{equation}
        where  $Q_m$ is the (finite) support of $U_m$, which also coincides with the support of $U_m^*$. The term $ h_m $ is self-adjoint and bounded, cf.~\eqref{eq: sufrace energy limit}. It represents the total interaction associated with $ Q_m $. Removing this term results in a state $ \omega^{h_m} $, which is 
        tracial on $ \mathcal{U}_{Q_m} $ and hence invariant under $ \tau_m $. This is the content of
        \begin{lemma}\label{lem: helper}
    With $ \omega \in K_\beta(\alpha) $ and $ h_m $, $ m \in \mathbb{N} $, from \eqref{delta to hm}:
    \begin{equation} \label{eq:tracial}
    \omega^{h_m}=\omega_{\mathrm{tr}}\otimes(\omega^{h_m}|_{\mathcal{U}_{\mathcal{L}\setminus Q_m}}), \quad \mbox{where}\quad  \omega_\mathrm{tr}:\mathcal{U}_{Q_m}\ni A\mapsto \frac{\mathrm{tr}(A)}{\mathrm{tr}(\mathbb{I}_{\mathcal{U}_{Q_m}})} .
    \end{equation}
    Consequently, $  \omega^{h_m} =  \omega^{h_m} \circ \tau_m $.
\end{lemma}
\begin{proof}
By \cite[Thm. 5.4.4]{BratelliRobinson2} we know that $\omega^{h_m}\in K_\beta(\alpha^{h_m})$. Since the quasi-local algebra $\mathcal{U}$ is simple \cite[Cor. 2.6.19]{bratteliRobinson1},   $\omega^{h_m} $ is faithful \cite[Thm. 5.3.10]{BratelliRobinson2}. In particular,   $\omega^{h_m}(B) >0 $ for any positive, non-zero $B\in\mathcal{U}_{\mathcal{L}\setminus Q_m}$. Hence 
       \begin{equation}
				\omega^{h_m}_B:\mathcal{U}_X\ni A\mapsto \bigl(\omega^{h_m}(B)\bigr)^{-1}\omega^{h_m}(AB)
			\end{equation}
			is a well-defined state on $\mathcal{U}_{Q_m}$. We now consider $A_1 , A_2 \in \mathcal{U}_{Q_m} $. Then $\alpha^{h_m}_t(A_1)=A_1$ for all $t\in\mathbb{R}$ and,  since  this is constant, the identity analytically extends to $t\in\mathbb{C}$. Applying the KMS-condition for $ \omega^{h_m}$ with respect to $ \alpha^{h_m} $ we hence conclude
            \begin{align*}    \omega^{h_m}_B(A_1A_2)=&\bigl(\omega^{h_m}(B)\bigr)^{-1}\omega^{h_m}(A_1A_2B) =\bigl(\omega^{h_m}(B)\bigr)^{-1}\omega^{h_m}\bigl(A_2B\alpha^{h_m}_{i\beta}(A_1)\bigr) \notag \\ =&\bigl(\omega^{h_m}(B)\bigr)^{-1}\omega^{h_m}(A_2A_1B)
                =\omega^{h_m}_B(A_2A_1).
				\end{align*}
			For the last equality we have used that $[A_1,B]=0$ due to disjoint supports. Therefore, $\omega_B^{h_m}$ is a tracial state, and since $\vert Q_m\vert <\infty$, the unique such state is $\omega_{\mathrm{tr}}$ \cite[Lemma IV.4.1]{Simon+1993}.
		      Thus, for all positive $B\in\mathcal{U}_{\mathcal{L}\setminus Q_m}$ and all $A\in\mathcal{U}_{Q_m}$ we have 
            $$\omega^{h_m}(A\otimes B)=\omega^{h_m}(AB)=\omega_{\mathrm{tr}}(A)\omega^{h_m}(B).$$ 
            This extends to all $ B\in\mathcal{U}_{\mathcal{L}\setminus Q_m} $ by decompositions into positive parts, which completes the proof of~\eqref{eq:tracial}.

          Since $Q_m $ is the support of $ U_m $, the cyclicity of the trace and~\eqref{eq:tracial} imply $\omega^{h_m}\circ\tau_m=\omega^{h_m}$.
            \end{proof}
\medskip
\noindent
2.~With $ \omega $ replaced by $ \omega^{h_m} $ and $ W= - h_m $ from \eqref{delta to hm}. 
Let $(\mathcal{H}_m, \pi_m, \Omega_m)$ stand for the GNS-triple of the positive functional $\omega^{h_m} $. In this representation, the automorphism $ \tau_m $ is unitarily implemented by $V_m \coloneqq \pi_m(U_m) $ and the invariance $ \omega^{h_m} = \omega^{h_m}\circ \tau_m $ from Lemma~\ref{lem: helper} translates to $ V_m \Omega_m = \Omega_m $. Likewise, since $\alpha^{h_m}_t(U_m)=U_m$ for all $t\in\mathbb{R}$, the unitary time-evolution of the self-adjoint GNS-Hamiltonian $H_m:\mathcal{H}_m\supseteq D(H_m)\rightarrow \mathcal{H}_m$ leaves $ V_m $ invariant, 
            \begin{equation}
				\label{eq: V commutes with dynamics}
				V_m=\pi_m(U_m)=\pi(\alpha_t^{h_m}(U_m))=e^{itH_m}V_me^{-itH_m}\;\; \mathrm{for}\;\mathrm{all}\; t\in\mathbb{R}. 
			\end{equation}
            Moreover, $\pi_m(h_m)$ is bounded on $\mathcal{H}$ and thus 
			\begin{equation}
				\label{eq: V and pi}
				V_me^{it\pi_m(h_m)}V^*_m=\exp{\bigl(itV_m\pi_m(h_m)V^*_m\bigr)}=\exp{\bigl(it\pi_m(U_m h_m U_m^*)\bigr)}.
			\end{equation}
          Moreover, the Trotter product formula is applicable \cite{trotter1959product} and yields
				\begin{align*}
					V_m\exp{\bigl(it(H_m+\pi_m(h_m)\bigr)}\Omega_m 
					&=\lim_{n\rightarrow \infty}V_m\bigl(e^{\frac{it}{n}H_m}e^{\frac{it}{n}\pi_m(h_m)}\bigr)^n\Omega_m =\lim_{n\rightarrow \infty}V_m\bigl(e^{\frac{it}{n}H_m}e^{\frac{it}{n}\pi_m(h_m)}\bigr)^nV^*_m\Omega_m\\
					&=\lim_{n\rightarrow \infty}\bigl(e^{\frac{it}{n}H_m}e^{\frac{it}{n}\pi_m(U_mh_mU_m^*)}\bigr)^n\Omega_m=\exp{\bigl(it(H_m+\pi_m(U_mh_mU_m^*)\bigr)}\Omega_m.
				\end{align*}
            The last line is based on \eqref{eq: V commutes with dynamics} and \eqref{eq: V and pi}, and again the Trotter product formula.
			It is then standard procedure (cf.~\cite[Thm IV.5.5]{Simon+1993}) to use the KMS-condition to analytically continue this expression to $t=i\beta/2$ and obtain
			\begin{equation}
				\label{eq: commuting V past the exponential}
				V_me^{- \frac{\beta}{2}(H_m+\pi_m(h_m))}\Omega_m=e^{-\frac{\beta}{2}(H_m+\pi_m(U_mh_mU_m^*))}\Omega_m.
			\end{equation}
            With these preparations, we establish the crucial explicit form of the symmetry-transformed KMS-state.
            \begin{lemma}\label{lem:perturbo}
                With $ \omega \in K_\beta(\alpha) $ and $ h_m $, $ m \in \mathbb{N} $, from \eqref{delta to hm}:
    \begin{equation} \label{eq:perts}
        \omega \circ \tau_m = \omega^{h_m-U_mh_mU_m^*}, \qquad \omega \circ \tau_m^{-1} = \omega^{h_m-U_m^*h_mU_m} .
    \end{equation}
            \end{lemma}
            \begin{proof}
			Writing $\omega=(\omega^{h_m})^{-h_m}$ and using the definition~\eqref{pertubed dynamics} of the perturbed functional, we conclude for all $ A \in \mathcal{U} $
			\begin{align*}
			\omega\circ\tau_m(A)& =\langle e^{-\frac{\beta}{2}(H_m-\pi_m(-h_m))}\Omega_m, V^*_m\pi_m(A)V_me^{-\frac{\beta}{2}(H_m-\pi_m(-h_m))}\Omega_m\rangle\\
					&=\Bigl\langle \exp{\Bigl(-\frac{\beta}{2}\bigl(H_m-\pi_m(-U_mh_mU_m^*)\bigr)\Bigr)}\Omega_m , \pi_m(A)\exp{\Bigl(-\frac{\beta}{2}\bigl(H_m-\pi_m(-U_mh_mU_m^*)\bigr)\Bigr)}\Omega_m\Bigr\rangle\\
					&=(\omega^{h_m})^{-U_mh_mU_m^*}(A) =\omega^{h_m-U_mh_mU_m^*}(A).
				\end{align*}
            The second equality followed from \eqref{eq: commuting V past the exponential} and the two remaining equalities are applications of \eqref{pertubed dynamics} again. This proves the first equality in \eqref{eq:perts}. The second one follows similarly. 
            \end{proof}

        We are now in the position to finalize the argument for Proposition~\ref{th: Proposition 1}. 
        \begin{proof}[Proof of Proposition~\ref{th: Proposition 1}]
        From~\eqref{relative entropy as limit} combined with~\eqref{relative entropy bdd} we conclude
             \begin{align}
                   S(\omega\vert\omega\circ\tau_m)
                   &\leq S(\omega\vert\omega\circ\tau_m) + S(\omega\vert\omega\circ\tau_m^{-1}) \notag \\
					&=- \beta\omega(h_m-U_mh_mU_m^*)-\beta\omega(h_m-  U_m^*h_mU_m)
            \end{align}
           where the equality is due to~\eqref{eq:perts} and~\eqref{entropy to state}. Conversely, \eqref{eq: time derivative} implies
            \begin{align} \label{delta to hm0}
            -i\bigl(U_m\delta(U_m^*)+U_m^*\delta(U_m)\bigr) 
               &= \lim_{n\rightarrow \infty} \sum_{\Lambda \subseteq \Lambda_n } \left( [U_m [\phi(\Lambda), U_m^*] + U_m^* [\phi(\Lambda), U_m] \right) \notag \\
                &=\lim_{n\rightarrow \infty} \sum_{\substack{\Lambda \subseteq \Lambda_n : \\ \Lambda \cap Q_m \neq 0}} \left( U_m \phi(\Lambda) U_m^* -\phi(\Lambda) + U_m^* \phi(\Lambda) U_m -\phi(\Lambda) \right) \notag \\
                &=U_m h_m U_m^* -h_m +  U_m^* h_m U_m - h_m . 
            \end{align}
        Hence $S(\omega\vert\omega\circ\tau_m) \leq \sup_{m\in\mathbb{N}}\bigl\vert \omega\bigl(U_m\delta(U_m^*)+U_m^*\delta(U_m)\bigr)\bigr\vert $, which concludes the proof of~\eqref{eq: bound expectation value of U_m delta U_m^* not appendix}. The rest of the assertion follows from Lemma~\ref{lem:fact1}. 
        \end{proof}

        \section{Generalization: Mermin-Wagner on a slab} \label{sect:Slab}
		In this Appendix we prove a variant of our main theorem. As a motivation, let us first consider the following example
		\begin{align*}
			\mathcal{L}=\mathbb{Z}^l\times\{1,...,M_{l+1}\}\times...\times\{1,...,M_{d}\}
		\end{align*}
		for some $M_j\in\mathbb{N}$. Then the effective dimension is $ l $, i.e., for $m\geq \max_j M_j$, $$\vert \mathcal{L}\cap[-m,m]^d\vert =(2m+1)^l\cdot \prod_{j=1}^{d-l} M_{l+j} , $$ 
		and therefore $\tau^{(a)}$ cannot be spontaneously broken for $l\leq2(k-\vert a\vert+1)$ if the other assumptions of Theorem~\ref{Thm:Main} hold. This illustrates that the $l$-dimensional "slab" $\mathbb{Z}^l\times\{1,...,M_{l+1}\}\times...\times\{1,...,M_{d}\}$ in $\mathbb{R}^d$ with finite thickness in the other $d-l$ spatial directions behaves effectively like an $l$-dimensional system for the Mermin-Wagner theorem. However, Assumption~\ref{as: interaction} of Theorem \ref{Thm:Main} still needs the symmetries $\tau^{(a)}$ for $a\in\mathbb{N}_0^d$ and not just $\mathbb{N}_0^l \times \{0\}^{d-l}$. 
		We can however modify the Mermin-Wagner theorem for the spontaneous breaking of $\tau^{(a)}$ with $b_{l+1}=...=b_d=0$ so to only require the presence of the multipole symmetries $\tau^{(a)}$ with $a_{l+1}=...=a_d=0$. We stress that one does not require that $\mathcal{L}$ can be realized as a Cartesian product of a square lattice and a finite set. Rather we have the following variant of Theorem~\ref{Thm:Main}.
		\begin{theorem}
			\label{th: second Mermin Wagner}
			 Suppose there is a subset $I\subseteq\{1,...,d\}$ and constants $M_j\geq0$ for $j\in I$ such that $\vert x_{j}\vert\leq \vert M_j\vert$ for all $x\in\mathcal{L}$ and all $j\in I$. Fix $k\in \mathbb{N}_0$ and assume that Assumptions \ref{as: lattice}, \ref{as: family of charge operators} hold. Furthemore, assume that Assumption \ref{as: interaction} hold, with \eqref{as:invariance} valid only for every multi-index $a\in \mathbb{N}_0^d$ with $\vert a\vert\leq k$ and $a_j=0$ for every $j\in I$. 
			 
			 Fix some inverse temperature $\beta\in \mathbb{R}$. If $a\in \mathbb{N}_0^d$ is a fixed multi-index with $\vert a\vert\leq k$ and $\gamma \leq 2(k-\vert a \vert +1)$, then the multipole symmetry $(\tau_s^{(a)})_{s\in\mathbb{R}}$ is preserved on every $\beta$-KMS state $\omega$ for $\beta\in\mathbb{R}$; that is $\omega\circ\tau_s^{(a)}=\omega$ for all $s\in\mathbb{R}$ and $\omega\in K_\beta(\alpha)$. 
		\end{theorem}
		
		\begin{proof}
			Theorem \ref{th: second Mermin Wagner} can be proven in the same fashion as Theorem \ref{Thm:Main} by combining Proposition \ref{th: Proposition 1} and the second part of Proposition \ref{th: Proposition main Taylor construction}. 
		\end{proof}
		
		Since we do not require any translational invariance of $\mathcal{L}$, Theorems \ref{Thm:Main} and \ref{th: second Mermin Wagner} can in particular be applied to lattice models of Moiré materials or similar constructions where one stacks several 2D "slices" on top of each other in a not necessarily translationally invariant way.

		\paragraph{Acknowledgements.} 
		This work was supported by the DFG under grant TRR 352--Project-ID 470903074. SW was also supported under EXC-2111 -- 390814868.

		\bibliography{References}
		\bibliographystyle{plain}
		
	\end{document}